\documentclass{amsart}
\usepackage{amsfonts}
\usepackage{blkarray}
\usepackage{multirow}

\newtheorem{theorem}{Theorem}[section]
\newtheorem{lemma}[theorem]{Lemma}
\newtheorem{proposition}[theorem]{Proposition}
\newtheorem{corollary}[theorem]{Corollary}
\theoremstyle{definition}
\newtheorem{definition}[theorem]{Definition}
\newtheorem{example}[theorem]{Example}

\theoremstyle{remark}

\numberwithin{equation}{section}

\begin{document}
\title[QR and QDC codes over $\mathbb{F}_{2}+u\mathbb{F}_{2}+u^{2}\mathbb{F}%
_{2}$ ]{New extremal binary self-dual codes of length $68$ from quadratic residue codes
over $\mathbb{F}_{2}+u\mathbb{F}_{2}+u^{2}\mathbb{F}_{2}$}
\author{Abidin Kaya}
\author{Bahattin Yildiz}
\author{Irfan Siap}
\address{Department of Mathematics, Fatih University, 34500, Istanbul, Turkey%
}
\email{byildiz@fatih.edu.tr, akaya@fatih.edu.tr}
\address{Department of Mathematics, Yildiz Technical University, 34210,
Istanbul, Turkey}
\email{isiap@yildiz.edu.tr}

\begin{abstract}
In this work, quadratic reside codes over the ring $\mathbb{F}_{2}+u\mathbb{F%
}_{2}+u^{2}\mathbb{F}_{2}$ with $u^{3}=u$ are considered. A duality and
distance preserving Gray map from $\mathbb{F}_{2}+u\mathbb{F}_{2}+u^{2}%
\mathbb{F}_{2}$ to $\mathbb{F}_2^3$ is defined. By using quadratic double
circulant, quadratic bordered double circulant constructions and their
extensions self-dual codes of different lengths are obtained. As Gray images
of these codes and their extensions, a substantial number of new extremal
self-dual binary codes are found. More precisely, thirty two new extremal
binary self-dual codes of length $68$, 363 Type I codes of parameters $%
[72,36,12]$, a Type II $[72,36,12]$ code and a Type II $[96,48,16]$ code
with new weight enumerators are obtained through these constructions. The
results are tabulated.
\end{abstract}

\maketitle

\section{Introduction}

Quadratic residue codes are a special family of BCH codes, which is a
special subfamily of cyclic codes. They were first introduced by Andrew
Gleason and since then have generated a lot of interest. This is due to the
fact that they enjoy good properties and they are source of good codes such
as binary quadratic residue codes. While being studied over finite fields in
the early works, recently quadratic residue codes have been studied over
some special rings.

First, Pless and Qian studied quaternary quadratic residue codes (over the
ring $\mathbb{Z}_{4}$) and some of their properties in \cite{Pless}. In
2000, Chiu et al. extended the ideas in \cite{Pless} to the ring $\mathbb{Z}%
_{8}$ in \cite{chiu}. Taeri considered quadratic residue codes over the ring
$\mathbb{Z}_{9}$ in \cite{Taeri}. Most recently, the authors studied
quadratic residue codes over the ring $\mathbb{F}_{p}+v\mathbb{F}_{p}$ and
their images in \cite{FpvFp}.

Another interesting and oft-studied class of codes is the class of self-dual
codes. Self-dual codes have connections to many fields of research such as lattices,
designs and invariant theory. The study of extremal self-dual codes has
generated a lot of interest among coding theorists. There are many different
constructions for them. We can direct the reader to see \cite%
{bouyukliev,dontcheva,dontcheva2,dougherty1,feit,gulliver} and the
references therein for a complete literature on self-dual codes.

The connection between quadratic residue codes and self-dual codes was first
explored quite effectively by Pless in seventies in constructing the
extremal doubly-even self dual code of parameters $[48,24,12]$. This code is
still known as the extended quadratic residue code. Gaborit used a quadratic
residue double circulant construction for self-dual codes in \cite{gaborit}.
In \cite{FpvFp}, the authors explored this connection using quadratic
residue codes over the ring $\mathbb{F}_{p}+v\mathbb{F}_{p}$ and constructed
a number of good self-dual codes over different fields.

Our goal in this work is to construct quadratic residue codes over a newly
defined ring $R = \mathbb{F}_2+u\mathbb{F}_2+u^2\mathbb{F}_2$ with $u^3=u$
and to explore new constructions for binary self-dual codes. A duality and
weight-preserving Gray map from the ring to the binary field allows us to
construct many good binary self-dual codes as Gray images of self-dual codes
over $R$.

The rest of the paper is organized as follows. In Section $2$, the structure
of the ring as well as some preliminaries about self-dual codes are given.
Quadratic residue codes and extended quadratic residue codes are defined and
investigated in Section $3$. Some extremal binary self dual codes are
obtained as Gray images. Particularly a Type II $\left[ 96,48,16\right] _{2}$
code with a new weight enumerator appeared in the examples. In Section $4$,
quadratic double circulant (QDC) and bordered QDC codes over $\mathbb{F}%
_{2}+u\mathbb{F}_{2}+u^{2}\mathbb{F}_{2}$ are defined. Families of self dual
codes are obtained. The Gray image of an example turned to be a type II $%
\left[ 72,36,12\right] _{2}$ binary code with a new weight enumerator. Some
extension methods for self dual codes over $R$ are given in Section $5$. As
a result, $363$ new $\left[ 72,36,12\right] _{2}$ Type I codes and $32$ new
extremal binary self-dual codes of parameters $[68,34,12]$ are obtained via
the Gray images of $R$-extensions. Section $6$ concludes the paper.

\section{Preliminaries}

\subsection{The structure of the ring $\mathbb{F}_{2}+u\mathbb{F}_{2}+u^{2}%
\mathbb{F}_{2}$ with $u^{3}=u$}

\bigskip Throughout, we let $R$ denote the commutative ring $\mathbb{F}_{2}+u%
\mathbb{F}_{2}+u^{2}\mathbb{F}_{2}$, constructed via $u^{3}=u$. $R$ is a
characteristic $2$ ring of size $8$. It is a non-local, non-chain principal
ideal ring with the following non-trivial ideals;%
\begin{eqnarray*}
I_{1+u} &=&\left( 1+u\right) =\left\{ 0,1+u,\ u+u^{2},1+u^{2}\right\} , \\
I_{u^{2}} &=&\left( u^{2}\right) =\left\{ 0,u,u^{2},u+u^{2}\right\} , \\
I_{u+u^{2}} &=&\left( u+u^{2}\right) =\left\{ 0,u+u^{2}\right\} , \\
I_{1+u^{2}} &=&\left( 1+u^{2}\right) =\left\{ 0,1+u^{2}\right\} ,
\end{eqnarray*}%
which satisfy $0\subset I_{1+u^{2}},\ I_{u+u^{2}}\subset I_{1+u},\ I_{u^{2}}$
$\subset R$.

The units in $R$ are given by $\left\{ 1,1+u+u^{2}\right\} $ and the square
of a unit is $1$. The non-units are given by $\left\{
0,u,u^{2},u+u^{2},1+u,1+u^{2}\right\} $ and splitted into three groups with
respect to their squares as
\begin{eqnarray*}
u^{2} &=&\left( u^{2}\right) ^{2}=u^{2}, \\
\left( 1+u\right) ^{2} &=&\left( 1+u^{2}\right) ^{2}=1+u^{2}, \\
0^{2} &=&\left( u+u^{2}\right) ^{2}=0.
\end{eqnarray*}%
The ring has primitive idempotents in $u^{2}$ and $1+u^{2}$. Note that the ring
is isomorphic to $\mathbb{F}_{2}\times \left( \mathbb{F}_{2}+u\mathbb{F}%
_{2}\right) $ if we label $u+u^{2}$ as $u$. Every element of $R$ can be
written uniquely in the form $\left( 1+u^{2}\right) a+u^{2}\left( b+c\left(
u+u^{2}\right) \right) $ where $a,b$ and $c\in \mathbb{F}_{2}$.

We introduce the character $\chi $ from the additive group of $R$ to nonzero
complex numbers as $\chi \left( a+bu+cu^{2}\right) =\left( -1\right) ^{c}$.
It is clear that $\chi (\alpha +\beta )=\chi (\alpha ).\chi (\beta )$ for
all $\alpha ,\beta \in R$. $ker(\chi )=\{0,1,u,1+u\}$, which does not
contain any non-trivial ideals of $R$. Thus by \cite{Wood}, we see that $%
\chi $ is a generating character of the ring. Since it has a generating
character, it is a Frobenius ring. In particular this means we have the
following lemma:

\begin{lemma}
\label{dual} Let $C$ be a linear code over $R$ of length $n$. Then $|C|\cdot
|C^{\perp }|=|R|^{n}=8^{n}.$
\end{lemma}

\subsection{Linear codes over $R$}

A linear code $C$ of length $n$ over $R$ is an $R$-submodule of $R^{n}$. An
element of the code $C$ is called a codeword of $C$. A generator matrix of $%
C $ is a matrix whose rows generate $C$. The Hamming
weight of a codeword is the number of non-zero components.

Let $x=\left( x_{1},x_{2},\ldots ,x_{n}\right) $ and $y=\left(
y_{1},y_{2},\ldots ,y_{n}\right) $ be two elements of $R^{n}$. The Euclidean
inner product is given as $\left\langle x,y\right\rangle _{E}=\sum
x_{i}y_{i} $. The dual code of $C$ with respect to the Euclidean inner
product is denoted by $C^{\bot }$ and defined as
\begin{equation*}
C^{\bot }=\left\{ x\in R^{n}\mid \left\langle x,y\right\rangle _{E}=0\text{
for all }y\in C\right\}
\end{equation*}%
We say that $C$ is self-dual if $C=C^{\bot }$.

Two linear codes are said to be permutation equivalent if one can be
obtained from the other by a permutation of coordinates. A code is said to
be iso-dual if it is permutation equivalent to its dual code.

In the sequel we let $R_{n}:=R\left[ x\right] /\left( x^{n}-1\right) $. A
polynomial $f\left( x\right) $ is abbreviated as $f$ if there is no
confusion.

The extended code of a code $C$ over $R$ will be denoted by $\overline{C}$,
which is the code obtained by adding a specific column to the generator
matrix of $C$.

Let $p$ be an odd prime such that $p\equiv \pm 1\pmod{8}$ and let $Q_{p}$
and $N_{p}$ be the sets of quadratic residues and non-residues modulo $p$,
respectively. We use the notations $e_{1}(x)=\sum\limits_{i\in Q_{p}}x^{i}$,
$e_{2}(x)=\sum\limits_{i\in N_{p}}x^{i}$ and $h$ denotes the polynomial
corresponding to the all one vector of length $p$, i.e. $h=1+e_{1}+e_{2}$.

Let $a\in \mathbb{F}_{p}^{\ast }$, the map $\mu _{a}:\mathbb{F}%
_{p}\rightarrow \mathbb{F}_{p}$ is defined by $\mu _{a}\left( i\right) =ai%
\pmod{p}$ and it acts on polynomials as
\begin{equation*}
\mu _{a}\left( \sum_{i}x^{i}\right) =\sum_{i}x^{\mu _{a}(i)}.
\end{equation*}%
It is easily observed that $\mu _{a}\left( fg\right) =\mu _{a}\left(
f\right) \mu _{a}\left( g\right) $ for polynomials $f$ and $g$ in $R_{p}$.

Let $S$ be a commutative ring with identity, then;

\begin{theorem}
\cite{Huffman}\cite{Taeri} Let $C_{1}$ and $C_{2}$ be cyclic codes of length
$n$ over $S$ generated by the idempotents $a,b$ in $S\left[ x\right] /\left(
x^{n}-1\right) $. Then $C_{1}\cap C_{2}$ and $C_{1}+C_{2}$ are generated by
the idempotents $ab$ and $a+b-ab$, respectively.
\end{theorem}

\begin{theorem}
\cite{Huffman}\cite{Taeri} Let $C$ be a cyclic code over $S$ generated by
idempotent $e\left( x\right) $. Then its dual $C^{\perp }$ is generated by
the idempotent $1-e\left( x^{-1}\right) $.
\end{theorem}

It is well-known that cyclic codes over $R$ correspond to ideals in $R_n =
R[x]/(x^n-1)$. Thus it is essential to understand the structure of the ring $%
R_n$. We observe that every element in $R_{n}$ can be written uniquely in
the form $\left( 1+u^{2}\right) f+u^{2}\left( g+h\left( u+u^{2}\right)
\right) $ where $f,g$ and $h\in \mathbb{F}_{2}\left[ x\right] /\left(
x^{n}-1\right) $. An important tool in studying the ring $R_n$ is to
consider the idempotents. We first show that the idempotents in $R_{n}$ are
characterized as follows:

\begin{lemma}
$\left( 1+u^{2}\right) f+u^{2}\left( g+h\left( u+u^{2}\right) \right) $ is
an idempotent in $R_{n}$ if and only if $f$ and $g$ are idempotents in $%
\mathbb{F}_{2}\left[ x\right] /\left( x^{n}-1\right) $ and $h$ is the zero
polynomial.
\end{lemma}

\begin{proof}
Let $\left( 1+u^{2}\right) f+u^{2}\left( g+h\left( u+u^{2}\right) \right) $
be an idempotent in $R_{n}$ then,%
\begin{eqnarray*}
\left[ \left( 1+u^{2}\right) f+u^{2}\left( g+h\left( u+u^{2}\right) \right) %
\right] ^{2} &=&\left( 1+u^{2}\right) f^{2}+u^{2}g^{2}\text{ since }\left(
u+u^{2}\right) ^{2}=0 \\
&=&\left( 1+u^{2}\right) f+u^{2}\left( g+h\left( u+u^{2}\right) \right)
\end{eqnarray*}%
implies $f^{2}=f$ and $g^{2}=g$.

Conversely, if $f$ and $g$ are idempotents then so is $\left( 1+u^{2}\right)
f+u^{2}g$.
\end{proof}

We define the following linear Gray map which takes a linear code over $R$
of length $n$ to a binary\ linear code of length $3n$.

\begin{definition}
\label{gray}Let $\varphi :R^{n}\rightarrow \mathbb{F}_{2}^{3n}$ be the map
given by%
\begin{equation*}
\varphi \left( \overline{a}+\overline{b}u+\overline{c}u^{2}\right) =\left(
\overline{a}+\overline{b},\overline{b}+\overline{c},\overline{c}\right),
\end{equation*}%
and define the Lee weight of an element of $R$ as $w_{L}\left(
a+bu+cu^{2}\right) =w_{H}\left( a+b,b+c,c\right) $ where $w_{H}$ denotes the
usual Hamming weight.
\end{definition}

\begin{proposition}
\label{duality}The Gray image of a self-dual code of length $n$ over $R$ is
a binary self-dual code of length $3n$.
\end{proposition}

\begin{proof}
We first show that the Gray images of orthogonal vectors in $R$ are
orthogonal in $\mathbb{F}_2$. Let $\overline{a}+\overline{b}u+\overline{c}%
u^{2}$ and $\overline{d}+\overline{e}u+\overline{f}u^{2}$ where $\overline{a}%
,\overline{b},\overline{c},\overline{d},\overline{e}$ and $\overline{f}\in
\mathbb{F}_{2}^{n}$ be two codewords of length $n$ over $R$ such that they
are orthogonal. Then
\begin{equation*}
\left \langle \overline{a}+\overline{b}u+\overline{c}u^{2}, \overline{d }+%
\overline{e}u+\overline{f}u^{2}\right \rangle = 0,
\end{equation*}
and so we get
\begin{equation*}
\overline{a}\overline{d}+\left(\overline{b}\overline{d}+\overline{a}%
\overline{e}+\overline{c}\overline{e}+\overline{b}\overline{f} \right)u+
\left( \overline{c}\overline{d}+\overline{b}\overline{e}+\overline{a}%
\overline{f}+\overline{c}\overline{f}\right)u^2 = 0.
\end{equation*}
This implies
\begin{equation}  \label{inpr}
\overline{a}\overline{d} = \overline{b}\overline{d}+\overline{a}\overline{e}+%
\overline{c}\overline{e}+\overline{b}\overline{f} = \overline{c}\overline{d}+%
\overline{b}\overline{e}+\overline{a}\overline{f}+\overline{c}\overline{f} =
0.
\end{equation}

Now, consider the inner product of the Gray images;
\begin{eqnarray*}
\left( \overline{a}+\overline{b},\overline{b}+\overline{c},\overline{c}%
\right)\cdot\left( \overline{d}+\overline{e },\overline{e}+\overline{f},%
\overline{f}\right) &=&\overline{a}\overline{d}+\overline{a}\overline{e}+%
\overline{b}\overline{d}+\overline{b}\overline{e}+\overline{b}\overline{e}+%
\overline{b}\overline{f}+ \overline{c}\overline{e}+\overline{c}\overline{f}+%
\overline{c}\overline{f} \\
&=&\overline{a}\overline{d}+\overline{a}\overline{e}+\overline{b}\overline{d}%
+\overline{b}\overline{f}+\overline{c}\overline{e}=0\text{,}
\end{eqnarray*}%
by (\ref{inpr}). This shows that
\begin{equation}  \label{subset}
\varphi(C^{\perp}) \subset \varphi(C)^{\perp}.
\end{equation}
But since $\varphi$ is an injective isometry, we have $|\varphi(C)| = |C|$.
Both $\mathbb{F}_2$ and $R$ are Frobenius, so we have
\begin{equation*}
|\varphi(C^{\perp})| = |C^{\perp}| = \frac{2^{3n}}{|C|} = \frac{8^n}{%
|\varphi(C)|} = |\varphi(C)^{\perp}|.
\end{equation*}
Combining this with (\ref{subset}), we conclude that $\varphi(C^{\perp})=%
\varphi(C)^{\perp}$. In particular this implies that the Gray images of
self-dual codes over $R$ are self-dual binary codes.
\end{proof}

A self-dual code over $R$ is said to be of Type II if the Lee weights of all
codewords are divisible by $4$, otherwise it is said to be of Type I. The
following corollary is an important consequence of Proposition \ref{duality}
and the definition of the Gray map:

\begin{corollary}
Suppose that $C$ is a self-dual code over $R$ of length $2n$ and minimum Lee
distance $d$. Then $\varphi(C)$ is a binary self-dual code of parameters $%
[6n,3n,d]$, and moreover $C$ and $\varphi(C)$ have the same weight
enumerators. In particular if $C$ is Type II(Type I), then so is $\varphi(C)$%
.
\end{corollary}

For binary self-dual codes we have the following upper bounds on the minimum
distances:

\begin{theorem}
(\cite{Rains}) Let $d_I(n)$ and $d_{II}(n)$ be the minimum distance of a
Type I and Type II binary code of length $n$. then
\begin{equation*}
d_{II}(n) \leq 4 \lfloor\frac{n}{24}\rfloor+4
\end{equation*}
and
\begin{equation*}
d_{I}(n) \leq \left \{
\begin{array}{ll}
4 \lfloor\frac{n}{24}\rfloor+4 & \text{if $n \not \equiv 22 \pmod{24}$} \\
4 \lfloor\frac{n}{24}\rfloor+6 & \text{if $n \equiv 22 \pmod{24}$.}%
\end{array}%
\right.
\end{equation*}
\end{theorem}

Self-dual codes meeting these bounds are called \textit{extremal}.

\section{Quadratic Residue Codes over $R$}

In this section, quadratic residue codes over the ring $R$ are defined in terms
of their idempotent generators. Extended and subtracted QR codes are also
defined. These codes and their Gray images are investigated. Codes with good
parameters are given as examples. In particular, the Gray image of the
extended quadratic residue code for $p=31$ turned out to be a type II $\left[
96,48,16\right] $ code, with a weight enumerator that was not known to exist before.

\begin{definition}
\label{defQR} Let $p$ be a prime such that $2$ is a quadratic residue modulo
$p$. Set $Q_{1}=\left\langle \left( 1+u^{2}\right) a+\left( u^{2}\right)
b\right\rangle $, $Q_{2}=\left\langle \left( 1+u^{2}\right) b+\left(
u^{2}\right) a\right\rangle $ and $Q_{1}^{\prime }=\left\langle \left(
1+u^{2}\right) a^{\prime }+\left( u^{2}\right) b^{\prime }\right\rangle $, $%
Q_{2}^{\prime }=\left\langle \left( 1+u^{2}\right) b^{\prime }+\left(
u^{2}\right) a^{\prime }\right\rangle $ where $a=e_{1}$, $b=e_{2}$, $%
a^{\prime }=1+e_{2}$ and $b^{\prime }=1+e_{1}$ if $p=8r-1$ and $a=1+e_{1}$, $%
b=1+e_{2}$, $a^{\prime }=$ $e_{2}$ and $b^{\prime }=e_{1}$ if $p=8r+1$.
These four codes are called quadratic residue codes over $R$ of length $p$.
\end{definition}

\begin{theorem}
\label{bQR} With the notation as in Definition \ref{defQR}, the following
hold for $R$-QR codes:

\item[a)] $Q_{1}$ and $Q_{1}^{\prime }$ are equivalent to $Q_{2}$ and $%
Q_{2}^{\prime }$, respectively;

\item[b)] $Q_{1}\cap Q_{2}=\left\langle h\right\rangle $ and $%
Q_{1}+Q_{2}=R_{p}$;

\item[c)] $\left\vert Q_{1}\right\vert =8^{\left( p+1\right) /2}=\left\vert
Q_{2}\right\vert ;$

\item[d)] $Q_{1}=Q_{1}^{\prime }+\left\langle h\right\rangle ,\
Q_{2}=Q_{2}^{\prime }+\left\langle h\right\rangle ;$

\item[e)] $\left\vert Q_{1}^{\prime }\right\vert =8^{\left( p-1\right)
/2}=\left\vert Q_{2}^{\prime }\right\vert ;$

\item[f)] $Q_{1}^{\prime }\cap Q_{2}^{\prime }=\left\{ 0\right\} $ and $%
Q_{1}^{\prime }+Q_{2}^{\prime }=\left\langle 1+h\right\rangle .$
\end{theorem}

\begin{proof}
The proof is an $R$-analogue of the proof of Theorem $3.2$ in \cite{FpvFp}.
We give the proof for the sake of completeness. Let $n\in N_{p}$ then $\mu
_{n}a=b$ and $\mu _{n}a^{\prime }=b^{\prime }$ therefore
\begin{equation*}
\mu _{n}\left[ \left( 1+u^{2}\right) a+u^{2}b\right] =\left( 1+u^{2}\right)
b+u^{2}a
\end{equation*}%
so $Q_{1}$ and $Q_{2}$ are equivalent. Similarly,
\begin{equation*}
\mu _{n}\left[ \left( 1+u^{2}\right) a^{\prime }+u^{2}b^{\prime }\right]
=\left( 1+u^{2}\right) b^{\prime }+u^{2}a^{\prime }
\end{equation*}%
therefore $Q_{1}^{\prime }$ and $Q_{2}^{\prime }$ are also equivalent.

\item[b)] $Q_{1}\cap Q_{2}$ is generated by the idempotent%
\begin{equation*}
\left[ \left( 1+u^{2}\right) a+u^{2}b\right] \left[ \left( 1+u^{2}\right)
b+va\right] =\left( 1+u^{2}\right) ab+u^{2}ab=ab=h
\end{equation*}%
and $Q_{1}+Q_{2}$ is generated by $\left( 1+u^{2}\right) a+u^{2}b+\left(
1+u^{2}\right) b+u^{2}a-ab=a+b-ab=1$.

\item[c)] From above it follows that
\begin{equation*}
\left( 2^{3}\right) ^{p}=\left\vert Q_{1}+Q_{2}\right\vert =\frac{\left\vert
Q_{1}\right\vert \left\vert Q_{2}\right\vert }{\left\vert Q_{1}\cap
Q_{2}\right\vert }=\frac{\left\vert Q_{1}\right\vert ^{2}}{2^{3}},
\end{equation*}%
so $\left\vert Q_{1}\right\vert =\left\vert Q_{2}\right\vert =8^{\left(
p+1\right) /2}$.

\item[d)] $Q_{1}^{\prime }+\left\langle h\right\rangle $ is generated by the
idempotent
\begin{eqnarray*}
&&\left( 1+u^{2}\right) a^{\prime }+u^{2}b^{\prime }+h-\left[ \left(
1+u^{2}\right) a^{\prime }+u^{2}b^{\prime }\right] h \\
&=&\left( 1+u^{2}\right) \left[ a^{\prime }+h-a^{\prime }h\right] +u^{2}%
\left[ b^{\prime }+h-b^{\prime }h\right] =\left( 1+u^{2}\right) a+u^{2}b
\end{eqnarray*}%
It follows that $Q_{1}^{\prime }+\left\langle h\right\rangle =Q_{1}$. In a
similar way, $Q_{2}=Q_{2}^{\prime }+\left\langle h\right\rangle $.

\item[e)] $8^{\left( p+1\right) /2}=\left\vert Q_{1}\right\vert =\left\vert
Q_{1}^{\prime }+\left\langle h\right\rangle \right\vert =\left\vert
Q_{1}^{\prime }\right\vert \left\vert \left\langle h\right\rangle
\right\vert =8\left\vert Q_{1}^{\prime }\right\vert $. Thus $\left\vert
Q_{1}^{\prime }\right\vert =8^{\left( p-1\right) /2}$.

\item[f)] $Q_{1}^{\prime }\cap Q_{2}^{\prime }$ is generated by the
idempotent
\begin{eqnarray*}
&&\left[ \left( 1+u^{2}\right) a^{\prime }+vb^{\prime }\right] \left[ \left(
1+u^{2}\right) b^{\prime }+va^{\prime }\right] \\
&=&\left( 1+u^{2}\right) a^{\prime }b^{\prime }+va^{\prime }b^{\prime } \\
&=&a^{\prime }b^{\prime }=0.
\end{eqnarray*}%
$Q_{1}^{\prime }+Q_{2}^{\prime }$ is generated by $\left( 1+u^{2}\right)
a^{\prime }+u^{2}b^{\prime }+\left( 1+u^{2}\right) b^{\prime
}+u^{2}a^{\prime }-0=a^{\prime }+b^{\prime }=1+h$.
\end{proof}

We define the extended QR codes over $R$ as follows;

\begin{definition}
For $i=1,2$ let $\overline{Q_{i}}$ be the code generated by the matrix
\begin{equation}
\overline{G_{i}}=\left(
\begin{array}{cccccc}
0 &  &  &  &  &  \\
0 &  &  & G_{i}^{\prime } &  &  \\
\vdots &  &  &  &  &  \\
1 & 1 & 1 & 1 & \cdots & 1%
\end{array}%
\right)  \label{gen}
\end{equation}%
where $G_{i}^{\prime }$ is a generating matrix of $Q_{i}^{\prime }$. $%
\overline{Q_{i}}$ are called extended QR codes over $R$.
\end{definition}

\begin{theorem}
\label{extend1}When $p\equiv -1\pmod{8}$ $Q_{1}^{\prime }$ and $%
Q_{2}^{\prime }$ are self-orthogonal and $Q_{1}^{\bot }=Q_{1}^{\prime }$ and
$Q_{2}^{\bot }=Q_{2}^{\prime }$. Moreover, $\overline{Q_{1}}$ and $\overline{%
Q_{2}}$ are self-dual codes.
\end{theorem}

\begin{proof}
$Q_{1}^{\bot }$ is generated by the idempotent $1-\left( \left(
1+u^{2}\right) e_{1}\left( x^{-1}\right) +u^{2}e_{2}\left( x^{-1}\right)
\right) $ and since $-1\in N_{p},\ e_{1}\left( x^{-1}\right) =e_{2}$ and $%
e_{2}\left( x^{-1}\right) =e_{1}$. So we have%
\begin{eqnarray*}
1-\left( \left( 1+u^{2}\right) e_{1}\left( x^{-1}\right) +u^{2}e_{2}\left(
x^{-1}\right) \right) &=&1+u^{2}+\left( 1+u^{2}\right) e_{1}\left(
x^{-1}\right) +u^{2}-u^{2}e_{2}\left( x^{-1}\right) \\
&=&\left( 1+u^{2}\right) \left[ 1-e_{1}\left( x^{-1}\right) \right] +u^{2}%
\left[ 1-e_{2}\left( x^{-1}\right) \right] \\
&=&\left( 1+u^{2}\right) \left( 1+e_{2}\right) +u^{2}\left( 1+e_{1}\right)
\end{eqnarray*}%
which implies $Q_{1}^{\bot }=Q_{1}^{\prime }$. Similarly, $Q_{2}^{\bot
}=Q_{2}^{\prime }$. So, $Q_{1}^{\prime }$ and $Q_{2}^{\prime }$ are
self-orthogonal since by Theorem \ref{bQR}, $Q_{1}^{\prime }\subset
Q_{1}=\left( Q_{1}^{\prime }\right) ^{\perp }$ and $Q_{2}^{\prime }\subset
Q_{2}=\left( Q_{2}^{\prime }\right) ^{\perp }$.

It is easily observed that for $i=1,2$ the rows of the matrix in $\left( \ref%
{gen}\right) $ are orthogonal to each other. Hence, $\overline{Q_{1}}$ and $%
\overline{Q_{2}}$ are self-dual codes.
\end{proof}

\begin{corollary}
The codes $\overline{Q_{1}}$ and $\overline{Q_{2}}$ are Type II codes and so
are the binary images.
\end{corollary}

\begin{proof}
Since $p\equiv -1\pmod{8}$ the weight of the last row of the matrix (\ref{gen})
has weight a multiple of $8$. The other rows have weight $w=1.1+\left( \frac{%
p-1}{2}\right) 2+\left( \frac{p-1}{2}\right) 3$ since weights of $1,u^{2}$
and $1+u^{2}$ are $1,2$ and $3$, respectively. If $p=8k-1$ then $w=4\left(
5k-1\right) $. Similarly, a row of the matrix \ref{gen} multiplied by $u$ or
$u^{2}$ has weight a multiple of $4$. So, $\overline{Q_{1}}$ and $\overline{%
Q_{2}}$ are type II codes. The result follows.
\end{proof}

\begin{theorem}
\label{extend2}When $p\equiv 1\pmod{8}$ $Q_{1}^{\bot }=Q_{2}^{\prime }$ and $%
Q_{2}^{\bot }=Q_{1}^{\prime }$. Furthermore, $\overline{Q_{1}}$ and $%
\overline{Q_{2}}$ are duals of each other.
\end{theorem}

\begin{proof}
$Q_{1}^{\bot }$ has idempotent generator
\begin{eqnarray*}
1-\left( \left( 1+u^{2}\right) \left( 1+e_{1}\left( x^{-1}\right) \right)
+u^{2}\left( 1+e_{2}\left( x^{-1}\right) \right) \right) &=&\left(
1+u^{2}\right) e_{1}\left( x^{-1}\right) +u^{2}e_{2}\left( x^{-1}\right) \\
&=&\left( 1+u^{2}\right) e_{1}\left( x\right) +u^{2}e_{2}\left( x\right)
\end{eqnarray*}%
which is the idempotent generator of $Q_{2}^{\prime }$. It follows that $%
Q_{1}^{\bot }=Q_{2}^{\prime }$. By similar steps we have $Q_{2}^{\bot
}=Q_{1}^{\prime }$.

$\left( Q_{2}^{\prime }\right) ^{\perp }=Q_{1}$ and $\left( Q_{1}^{\prime
}\right) ^{\perp }=Q_{2}$ it follows that the first $\frac{p-1}{2}$ rows of
matrix $\overline{G_{1}}$ are orthogonal to the first $\frac{p-1}{2}$ rows
of $\overline{G_{2}}$. All $1$ vector of length $p$ is in both $Q_{1}$ and $%
Q_{2}$ so it is in their dual spaces which implies the last rows of $%
\overline{G_{1}}$ and $\overline{G_{2}}$ are orthogonal to first $\frac{p-1}{%
2}$ rows of $\overline{G_{2}}$ and $\overline{G_{1}}$ respectively. It is
easily observed that the last rows of $\overline{G_{1}}$ and $\overline{G_{2}%
}$ are orthogonal. Hence, we have the result $\left( \overline{Q_{1}}\right)
^{\perp }=\overline{Q_{2}}$.
\end{proof}

Since the corresponding codes are equivalent, from now on we will use the
notations $QR^{\prime }\left( p\right) ,QR\left( p\right) $ and $\overline{%
QR\left( p\right) }$ for $Q_{1}^{\prime }$, $Q_{1}$ and $\overline{Q_{1}}$
respectively. By theorems \ref{extend1}, \ref{extend2} and proposition \ref%
{duality} we have the following result;

\begin{corollary}
$\overline{QR\left( p\right) }$ and its Gray image are self-dual codes when $%
p\equiv -1\pmod{8}$ and isodual codes when $p\equiv 1\pmod{8}$ .
\end{corollary}

For the case $p\equiv -1\pmod{8}$, we define the subtracted codes which are
Type I codes as follows;

\begin{definition}
The codes denoted by $SQR\left( p\right) $ and $BSQR\left( p\right) $ are
called subtracted and binary subtracted quadratic residue codes,
respectively and are defined as follows:
\begin{eqnarray*}
SQR\left( p\right) &=&\left\{ c\in R^{p-1}|(a,c,a)\in \overline{QR\left(
p\right) }\text{ for some }a\in R\right\} , \\
BSQR\left( p\right) &=&\left\{ c\in \mathbb{F}_{2}^{3p+1}|(a,c,a)\in \varphi
\left( \overline{QR\left( p\right) }\text{ }\right) \text{ for some }a\in
\mathbb{F}_{2}\right\} .
\end{eqnarray*}
\end{definition}

\begin{example}
For $p=7$ the odd-like quadratic residue code $QR^{\prime }$ has idempotent
generator $\left( 1+u^{2}\right) \left( 1+e_{2}\right) +\left( u^{2}\right)
\left( 1+e_{1}\right) =1+u^{2}e_{1}+\left( 1+u^{2}\right) e_{2}$ and the
code is self-orthogonal. So, $\overline{QR\left( 7\right) }$ is the code
generated by the matrix
\begin{equation*}
\left[
\begin{array}{cccccccc}
0 & 1 & u^{2} & u^{2} & 1+u^{2} & u^{2} & 1+u^{2} & 1+u^{2} \\
0 & 1+u^{2} & 1 & u^{2} & u^{2} & 1+u^{2} & u^{2} & 1+u^{2} \\
0 & 1+u^{2} & 1+u^{2} & 1 & u^{2} & u^{2} & 1+u^{2} & u^{2} \\
1 & 1 & 1 & 1 & 1 & 1 & 1 & 1%
\end{array}%
\right]
\end{equation*}%
and the binary Gray image of the code is the Golay code which is the unique
extremal Type II $[24,12,8]_{2}$ code. $\ SQR\left( 7\right) $ is the code
generated by
\begin{equation*}
\left[
\begin{array}{cccccc}
u^{2} & 1+u^{2} & 0 & 1 & 1 & 1 \\
1+u^{2} & 1 & 1 & u^{2} & 0 & 1 \\
1 & 1 & 1 & 1 & 1 & 1%
\end{array}%
\right]
\end{equation*}%
and the Gray image is an extremal $\left[ 18,9,4\right] _{2}$ code.
Similarly, binary subtracted code $BSQR\left( 7\right) $ is an extremal $%
\left[ 22,11,6\right] _{2}$ code.
\end{example}

\begin{example}
The Gray image of the code $\overline{QR\left( 23\right) }$ is a self-dual
Type II $\left[ 72,36,12\right] _{2}$ code and its weight enumerator has the
following form
\begin{equation*}
W_{72}=1+\left( 4398+\alpha \right) y^{12}+\left( 197073-12\alpha \right)
y^{16}+\cdots
\end{equation*}%
$\phi \left( \overline{QR\left( 23\right) }\right) $ has $\alpha =-1362$ and
$\left\vert Aut\right\vert =36432=2^{4}3^{2}11\times 23$. The code was
introduced in \cite{dontcheva2}. Moreover, the code $BSQR\left( 23\right) $
is an extremal $\left[ 70,35,12\right] _{2}$ self-dual code which was also
constructed in \cite{dontcheva2}.
\end{example}

As in the previous examples, the case for $p=31$ is also interesting. So, we
note it as the last example;

\begin{example}
A self-dual Type II $\left[ 96,48,16\right] _{2}$-code has weight enumerator%
\begin{equation*}
W_{96}=1+\left( -28086+\alpha \right) y^{16}+\left( 3666432-16\alpha \right)
y^{20}+\cdots
\end{equation*}%
The first such code with $\alpha =37722$ is constructed in \cite{feit} by a
construction from extended binary quadratic residue codes of length $32$ and
$25$ new codes are constructed in \cite{dontcheva} via automorphisms of
order $23$. In our case, the Gray image of $\overline{QR\left( 31\right) }$
has a weight enumerator with $\alpha =41106$ and $\left\vert Aut\right\vert
=89280=2^{6}3^{2}5\times 31$. A code with this weight enumerator was not
known previously. The binary subtracted quadratic residue code $BSQR\left(
31\right) $ is a $\left[ 94,47,14\right] _{2}$ code.
\end{example}

We finish this section by combining the results in the following tables:

\begin{center}
Table 1: QR codes for $p=8r-1$

\bigskip

\begin{tabular}{|l|l|l|}
\hline
& The code over $R$ & binary Gray image \\ \cline{2-3}
$QR^{\prime }\left( 7\right) $ & $\left( 7,8^{3},8\right) $ & $[21,9,8]_{2}$
\\ \cline{2-3}
$QR\left( 7\right) $ & $\left( 7,8^{4},5\right) $ & $[21,12,5]_{2}$ \\
\cline{2-3}
$\overline{QR\left( 7\right) }$ & $\left( 8,8^{4},6\right) $ & $%
[24,12,8]_{2} $ extremal self-dual \\ \cline{2-3}
$QR^{\prime }\left( 23\right) $ & $\left( 23,8^{11},12\right) $ & $%
[69,33,12]_{2}$ \\ \cline{2-3}
$QR\left( 23\right) $ & $\left( 23,8^{12},11\right) $ & $[69,36,11]_{2}$ \\
\cline{2-3}
$\overline{QR\left( 23\right) }$ & $\left( 24,8^{12},12\right) $ & $%
[72,36,12]_{2}$ Type II $\alpha =-1362$ in $W_{72}$ \\ \cline{2-3}
$QR^{\prime }\left( 31\right) $ & $\left( 31,8^{15},16\right) $ & $%
[93,45,16]_{2}$ \\ \cline{2-3}
$QR\left( 31\right) $ & $\left( 31,8^{16},14\right) $ & $[93,48,14]_{2}$ \\
\cline{2-3}
$\overline{QR\left( 31\right) }$ & $\left( 32,8^{16},16\right) $ & $%
[96,48,16]_{2}$ Type II $\alpha =41106$ in $W_{96}$ \\ \hline
\end{tabular}

\bigskip

Table 2: QR codes for $p=8r+1$

\bigskip

\begin{tabular}{|l|l|l|}
\hline
& The code over $R$ & binary Gray image \\ \cline{2-3}
$QR^{\prime }\left( 17\right) $ & $\left( 17,8^{8},10\right) $ & $%
[51,24,10]_{2}$ \\ \cline{2-3}
$QR\left( 17\right) $ & $\left( 17,8^{9},9\right) $ & $[51,27,9]_{2}$ \\
\cline{2-3}
$\overline{QR\left( 17\right) }$ & $\left( 18,8^{9},10\right) $ & $%
[54,27,10]_{2}$ \ isodual \\ \cline{2-3}
$QR^{\prime }\left( 41\right) $ & $\left( 41,8^{20},20\right) $ & $%
[123,60,20]_{2}$ \\ \cline{2-3}
$QR\left( 41\right) $ & $\left( 41,8^{21},18\right) $ & $[123,63,18]_{2}$ \\
\cline{2-3}
$\overline{QR\left( 41\right) }$ & $\left( 42,8^{21},20\right) $ & $%
[126,63,20]_{2}$ isodual \\ \hline
\end{tabular}
\end{center}

\section{Quadratic double circulant (QDC) and bordered QDC codes}

Quadratic double circulant codes are introduced in \cite{gaborit}. It is
observed that QDC codes is an important family of codes. In this section, we
define QDC codes over $R$ and obtain some families of self-dual codes. Some
extremal binary codes are obtained as Gray images of these codes. A Type II $%
[72,36,12]_{2}$ code with a new weight enumerator is obtained as the image
of a bordered QDC code for $p=11$.

Let $S$ be a commutative ring with identity, $r,s,t$ be elements of $S$, $v$
be the vector of length $p$ over $S$ and we label the $i$-th column by $%
i-1\in \mathbb{F}_{p}$ and define $i$-th entry of $v$ as $r$ if $i=1$, $s$
if $i-1$ is a quadratic residue in $\mathbb{F}_{p}$ and $t$ otherwise. Let $%
Q_{p}\left( r,s,t\right) $ be the $p\times p$ circulant matrix with the
first row $v$. Theorem 3.1 in \cite{gaborit} can be restated for the special
case where $q=p$ a prime and $char\left( S\right) =2$ as follows;

\begin{theorem}
\cite{gaborit}\label{qdc}Let $p$ be an odd prime and let $Q_{p}\left(
r,s,t\right) $ be a quadratic residue circulant matrix with $r,s$ and $t$
elements of the ring $S$. If $p=4k+1$ then%
\begin{eqnarray*}
&&Q_{p}\left( r,s,t\right) Q_{p}\left( r,s,t\right) ^{T} \\
&=&Q_{p}\left( r^{2},-s^{2}+k\left( s+t\right) ^{2},-t^{2}+k\left(
s+t\right) ^{2}\right)
\end{eqnarray*}%
If $p=4k+3$ then
\begin{eqnarray*}
&&Q_{p}\left( r,s,t\right) Q_{p}\left( r,s,t\right) ^{T} \\
&=&Q_{p}\left( r^{2}+s^{2}+t^{2},rs+rt+k\left( s+t\right)
^{2}+st,rs+rt+k\left( s+t\right) ^{2}+st\right) .
\end{eqnarray*}
\end{theorem}

In order to define quadratic double circulant and bordered quadratic double
circulant codes over $R$ we introduce the following matrices;

\begin{eqnarray*}
C_{p}\left( r,s,t\right) &=&\left[
\begin{array}{c|c}
I_{p} & Q_{p}\left( r,s,t\right)%
\end{array}%
\right] \\
B_{p}\left( r,s,t,\lambda ,\beta ,\gamma \right) &=&\left[
\begin{array}{ccccc|cccc}
& \multicolumn{3}{c}{\multirow{4}{*}{$I_{p+1}$}} &  & \lambda & \beta &
\cdots & \beta \\
&  &  &  &  & \gamma & \multicolumn{3}{c}{\multirow{3}{*}{$Q_{p}\left(
r,s,t\right)$}} \\
&  &  &  &  & \vdots &  &  &  \\
&  &  &  &  & \gamma &  &  &
\end{array}%
\right]
\end{eqnarray*}

\begin{definition}
The code generated by $C_{p}\left( r,s,t\right) $ over $R$ is called
quadratic double circulant code and is denoted by $\mathcal{C}_{p}\left(
r,s,t\right) $. In a similar way, the code generated by $B_{p}\left(
r,s,t,\lambda ,\beta ,\gamma \right) $ over $R$ is called bordered quadratic
double circulant code and is denoted by $\mathcal{B}_{p}\left( r,s,t,\lambda
,\beta ,\gamma \right) $.
\end{definition}

\begin{theorem}
The codes $\mathcal{C}_{p}\left( 0,u^{2},1+u^{2}\right) $ and $\mathcal{C}%
_{p}\left( u+u^{2},1+u,u\right) $ are self dual codes when $p\equiv 3\pmod{8}
$.
\end{theorem}

\begin{proof}
If $p\equiv 8k+3$ then by theorem \ref{qdc} $Q_{p}\left(
0,u^{2},1+u^{2}\right) Q_{p}\left( 0,u^{2},1+u^{2}\right) ^{T}=Q_{p}\left(
\left( u^{2}\right) ^{2}+\left( 1+u^{2}\right) ^{2},u^{2}\left(
1+u^{2}\right) ,u^{2}\left( 1+u^{2}\right) \right) =Q_{p}\left( 1,0,0\right)
=I_{p}$. Hence, the code $\mathcal{C}_{p}\left( 0,u^{2},1+u^{2}\right) $ is
a self dual code. Similarly for $\mathcal{C}_{p}\left( u+u^{2},1+u,u\right) $%
.
\end{proof}

Below, we list some good examples of this type:

\begin{center}
Table 3: Some self-dual double circulant codes

\bigskip

\begin{tabular}{|l|l|l|l|}
\hline
& The code over $R$ & binary Gray image & $\left\vert Aut\left( C\right)
\right\vert $ \\ \hline
$\mathcal{C}_{3}\left( 0,u^{2},1+u^{2}\right) $ & $\left( 6,8^{3},4\right) $
& $[18,9,4]_{2}$ &  \\ \hline
$\mathcal{C}_{3}\left( u+u^{2},1+u,u\right) $ & $\left( 6,8^{3},4\right) $ &
$[18,9,4]_{2}$ &  \\ \hline
$\mathcal{C}_{11}\left( 0,u^{2},1+u^{2}\right) $ & $\left(
22,8^{11},12\right) $ & $[66,33,12]_{2},$ $\alpha =66$ in $W_{66,1}$ & $660$
\\ \hline
$\mathcal{C}_{11}\left( u+u^{2},1+u,u\right) $ & $\left( 22,8^{11},12\right)
$ & $[66,33,12]_{2},$ $\alpha =22$ in $W_{66,1}$ & $220$ \\ \hline
$\mathcal{C}_{19}\left( 0,u^{2},1+u^{2}\right) $ & $\left(
38,8^{19},16\right) $ & $[114,57,16]_{2}$ &  \\ \hline
$\mathcal{C}_{19}\left( u+u^{2},1+u,u\right) $ & $\left( 38,8^{19},16\right)
$ & $[114,57,16]_{2}$ &  \\ \hline
\end{tabular}
\end{center}

Similar to the cases above, we may observe that $\mathcal{B}_{p}\left( r,s,t,\lambda ,1,1\right) $ is a
self dual code if $Q_{p}\left( r,s,t\right) Q_{p}\left( r,s,t\right)
^{T}=Q_{p}\left( 0,1,1\right) $ and the sum of the elements in a row of the
circulant matrix is $\lambda $ which satisfies $\lambda ^{2}=0$.

Some examples falling into this family are given below:

\begin{center}
Table 4: Some self-dual bordered double circulant codes

\bigskip

\begin{tabular}{|l|l|l|}
\hline
& binary Gray image & $\left\vert Aut\left( C\right) \right\vert $ \\ \hline
$\mathcal{B}_{11}\left( 1,u^{2},1+u^{2},0,1,1\right) $ & $[72,36,12]_{2},$ $%
\alpha =-3600$ in $W_{72}$ & $7920$ \\ \hline
$\mathcal{B}_{11}\left( u^{2},1,1+u^{2},0,1,1\right) $ & $[72,36,12]_{2},$ $%
\alpha =-1356$ in $W_{72}$ & $79200$ \\ \hline
$\mathcal{B}_{11}\left( u^{2},1,1+u,u+u^{2},1,1\right) $ & $[72,36,12]_{2}$
\ Type I & $440$ \\ \hline
$\mathcal{B}_{19}\left( 1,u^{2},1+u^{2},0,1,1\right) $ & $[120,60,16]_{2}$ &
\\ \hline
$\mathcal{B}_{19}\left( 1,u^{2},1+u,u+u^{2},1,1\right) $ & $[120,60,14]_{2}$
&  \\ \hline
\end{tabular}
\end{center}

The Gray image of the code $\mathcal{B}_{11}\left(
u^{2},1,1+u^{2},0,1,1\right) $ is the first $[72,36,12]_{2}~$Type II code
with a weight enumerator that has $\alpha =-1356$ in $W_{72}$, the binary generator
matrix is available online in \cite{kaya}. A code with $\alpha =-3600$ and $%
\left\vert Aut\right\vert =72$ is constucted in \cite{gulliver}, the code we
constructed with the same weight enumerator as the Gray image of $\mathcal{B}%
_{11}\left( 1,u^{2},1+u^{2},0,1,1\right) $ has an automorphism group of size
$7920$ which implies it is a new code.

\section{Extensions}

Some extension methods for self-dual codes are given and applied to some of
the codes in the previous section. In particular, we obtain $32$ new
extremal self-dual binary codes of length $68$, $363$ new Type I $\left[
72,36,12\right] _{2}$ codes, codes with these weight enumerators were not
known to exist previously.

In the sequel, let $S$ be a commutative ring of characteristic $2$ with
identity.

\begin{theorem}
\label{ext}Let $C$ be a self-dual code over $S$ of length $n$ and $G=(r_{i})$
be a $k\times n$ generator matrix for $C$, where $r_{i}$ is the $i$-th row
of $G$, $1\leq i\leq k$. Let $c$ be a unit in $S$ such that $c^{2}=1$ and $X$
be a vector in $R^{n}$ with $\left\langle X,X\right\rangle =1$. Let $%
y_{i}=\left\langle r_{i},X\right\rangle $ for $1\leq i\leq k$. Then the
following matrix%
\begin{equation*}
\left[
\begin{array}{ccc}
1 & 0 & X \\
y_{1} & cy_{1} & r_{1} \\
\vdots & \vdots & \vdots \\
y_{k} & cy_{k} & r_{k}%
\end{array}%
\right],
\end{equation*}
generates a self-dual code $D$ over $S$ of length $n+2$.
\end{theorem}

A quick search for the possible $R$-extensions of the codes $\mathcal{C}%
_{11}\left( 0,u^{2},1+u^{2}\right) $ and $\mathcal{C}_{11}\left(
u+u^{2},1+u,u\right) $ gave $16$ new $\left[ 72,36,12\right] _{2}$ codes
with known weight enumerators. In order to save space we do not list the
corresponding $\alpha $-value, $X$ and $c$ which are all available online in
\cite{kaya}.

A more specific extension method which can easily be applied to some double
circulant codes may be given as follows:

\begin{theorem}
\label{idext}Let $C$ be a self-dual code generated by $G=\left[ I_{n}|A%
\right] $ over $S$. If the sum of the elements in any row of $A$ is the unit
$v$ then the matrix:%
\begin{equation*}
G^{\ast }=\left[
\begin{array}{cc|cccccc}
1 & 0 & x_{1} & \ldots & x_{n} & v^{-1} & \ldots & v^{-1} \\ \hline
y_{1} & cy_{1} & \multicolumn{3}{c}{} & \multicolumn{3}{c}{} \\
\vdots & \vdots &  & I_{n} &  &  & A &  \\
y_{n} & cy_{n} &  &  &  &  &  &
\end{array}%
\right],
\end{equation*}
where $y_{i}=x_{i}+1$, $c$ is a unit with $c^{2}=1$, $\left\langle
X,X\right\rangle =1+nv^{-2}$ and $X=\left( x_{1},\ldots ,x_{n}\right) $,
generates a self-dual code $C^{\ast }$ over $S$.
\end{theorem}

\subsection{New Type I $\left[ 72,36,12\right] _{2}$ codes}

The existence of an extremal Type I $\left[ 72,36,14\right] _{2}$ code is
unknown. It is known that the non-existence of this code implies the
non-existence of the putative Type II $\left[ 72,36,16\right] _{2}$ code. So
far the best known distance for a Type I code of length $72$ is $12$ and few
such codes are known. See \cite{dougherty2} for some of them.

The possible weight enumerators for a Type I $\left[ 72,36,12\right] _{2}$
code are as follows;%
\begin{eqnarray*}
W_{72,1} &=&1+2\beta y^{12}+\left( 8640-64\gamma \right) y^{14}+\left(
124281-24\beta +384\gamma \right) y^{16}+\cdots \\
W_{72,2} &=&1+2\beta y^{12}+\left( 7616-64\gamma \right) y^{14}+\left(
134521-24\beta +384\gamma \right) y^{16}+\cdots
\end{eqnarray*}%
where $\beta $ and $\gamma $ are parameters. Observe that the three possible
weight enumerators for a $\left[ 72,36,14\right] _{2}$ code can be obtained
as $\beta =0=\gamma $ in $W_{72,2}$ and $\beta =0$ and $\gamma =0$, $1$ in $%
W_{72,1}$.

\begin{example}
The Type I code $\mathcal{B}_{11}\left( u^{2},1,1+u,u+u^{2},1,1\right) $ in
the previous section has weight enumerator $\gamma =11$ and $\beta =859$ in $%
W_{72,2}$.
\end{example}

\begin{example}
When we apply the extension in Theorem \ref{idext} to $\mathcal{C}%
_{11}\left( 0,u^{2},1+u^{2}\right) $ with $X=\left(
u^{2},0,u^{2},0,u^{2},u^{2},0,0,u+u^{2},u,u\right) $ and $c=1$, the Gray
image of the extension is a code with weight enumerator $\gamma =0$ and $%
\beta =335$ in $W_{72,2}$.
\end{example}

In a similar way codes with $\gamma =0$ and $\beta =209,263,309,317$ are
obtained from extensions of codes $\mathcal{C}_{11}\left(
u+u^{2},1+u,u\right) $ and $\mathcal{C}_{11}\left( 0,u^{2},1+u^{2}\right) $,
details are available in \cite{kaya}.

By considering the possible extensions of $BSQR\left( 23\right) $ with
respect to Theorem \ref{ext} we obtain $134$ self-dual codes of length $72$
with new weight enumerators in $W_{72,1}$. To be precise, the codes with $%
\gamma =0$ and $\beta =$523,...,575,577, 579, 580, $\gamma =1$ and $\beta =$%
525,\ 526,\ 527,\ 532,\ 533,\ 534,\ 539,...,577, 579,\ 580,\ 581, $\gamma =2$
and $\beta =$527,\ 538,\ 542,\ 549,\ 552,\ 555,\ 560,\ 562,\ 564,\ 565,\
566,\ 568,\ldots ,\ 573,\ 575,\ 576,\ 580,\ 584,\ 585$,$ $\gamma =3$ and $%
\beta =$548, 552, 558, 562, 568, 581, 582 and a code with $\gamma =4$ and $%
\beta =581$. The codes are available in \cite{kaya}.

The extension in Theorem \ref{idext} is applied to $\mathcal{C}_{11}\left(
0,u^{2},1+u^{2}\right) $ with $X$ and $c$, the Gray images of these codes
are self-dual codes of length $72$. Among them we single out Type I codes
with minimum distance $12$ and obtain $61$ different codes with weight
enumerators in $W_{72,1}$, here we list some of them;

\bigskip

\begin{center}
Table 5: Type I $\left[ 72,36,12\right] _{2}$ codes from $\mathcal{C}%
_{11}\left( 0,u^{2},1+u^{2}\right) $

\begin{tabular}{|l|l|l|l|}
\hline
$X$ & $c$ & $\gamma $ & $\beta $ \\ \hline
$\left( 1+u^{2},u,u,u,1+u^{2},u,0,1+u^{2},u,1+u^{2},u^{2}\right) $ & $1$ & $%
5 $ & $269$ \\ \hline
$\left( u+u^{2},0,u,1+u,u+u^{2},1,u^{2},1+u+u^{2},u,1+u^{2},u\right) $ & $1$
& $5$ & $273$ \\ \hline
$\left( u,1+u+u^{2},u^{2},u+u^{2},1+u,u^{2},1+u^{2},u+u^{2},0,1,u^{2}\right)
$ & $1$ & $5$ & $235$ \\ \hline
$\left( 1,1+u+u^{2},1+u,0,u,u^{2},u,0,u,u+u^{2},1+u\right) $ & $1+u+u^{2}$ &
$5$ & $255$ \\ \hline
$\left( u^{2},1+u^{2},u^{2},u+u^{2},u,u,1+u,1+u,1+u^{2},u^{2},u\right) $ & $%
1+u+u^{2}$ & $4$ & $263$ \\ \hline
$\left( 0,1,1,u^{2},u,1+u+u^{2},u+u^{2},1+u^{2},u,1+u^{2},1+u^{2}\right) $ &
$1+u+u^{2}$ & $3$ & $250$ \\ \hline
$\left( u^{2},0,1,u+u^{2},0,1,1+u^{2},1+u^{2},1+u,1,0\right) $ & $1$ & $3$ &
$258$ \\ \hline
$\left( 1+u^{2},1,0,u+u^{2},u,0,1+u+u^{2},1+u,u^{2},u^{2},u\right) $ & $%
1+u+u^{2}$ & $2$ & $279$ \\ \hline
$\left( 0,u,0,1+u^{2},1,1+u^{2},1+u^{2},u,1+u+u^{2},1,u\right) $ & $1$ & $1$
& $256$ \\ \hline
$\left( u,u^{2},u,1,1+u^{2},1,u+u^{2},0,1,1+u,1+u^{2}\right) $ & $1$ & $0$ &
$258$ \\ \hline
\end{tabular}
\end{center}

Same method is applied to $\mathcal{C}_{11}\left( u+u^{2},1+u,u\right) $ and
codes with $47$ distinct weight enumerators are obtained. Some of them are:

\bigskip

\begin{center}
Table 6: Type I $\left[ 72,36,12\right] _{2}$ codes from $\mathcal{C}%
_{11}\left( u+u^{2},1+u,u\right) $

\begin{tabular}{|l|l|l|l|}
\hline
$X$ & $c$ & $\gamma $ & $\beta $ \\ \hline
$\left( 1+u,u^{2},1+u^{2},u+u^{2},u^{2},u^{2},u^{2},u,u,1+u^{2},1+u\right) $
& $1+u+u^{2}$ & $4$ & $231$ \\ \hline
$\left( u^{2},u^{2},u,u^{2},u^{2},1+u^{2},1+u,1+u,u,1+u,0\right) $ & $%
1+u+u^{2}$ & $4$ & $249$ \\ \hline
$\left( 1+u,1+u,1+u+u^{2},1,0,0,u,1+u,u^{2},u^{2},1\right) $ & $1+u+u^{2}$ &
$3$ & $196$ \\ \hline
$\left(
u^{2},u^{2},1+u^{2},1+u^{2},1+u^{2},u^{2},u^{2},1+u^{2},0,u,u,u\right) $ & $%
1+u+u^{2}$ & $3$ & $215$ \\ \hline
$\left( 0,u^{2},1+u,1+u+u^{2},1+u+u^{2},1,0,1+u^{2},0,u,1\right) $ & $1$ & $%
2 $ & $241$ \\ \hline
$\left( u,1+u,1+u,1,u^{2},1+u^{2},0,1,1+u^{2},1,1+u+u^{2}\right) $ & $1$ & $%
2 $ & $244$ \\ \hline
$\left( u+u^{2},1,u+u^{2},1+u,u,u^{2},0,u+u^{2},0,1+u+u^{2},1+u\right) $ & $%
1+u+u^{2}$ & $2$ & $233$ \\ \hline
$\left( 1+u,0,0,1,1+u,0,1+u+u^{2},1,u,1+u+u^{2},u^{2}\right) $ & $1+u+u^{2}$
& $1$ & $211$ \\ \hline
$\left( 1,1,u^{2},1+u,u,1+u^{2},1,1+u^{2},0,u^{2},0\right) $ & $1$ & $1$ & $%
232$ \\ \hline
$\left( u,u,1,1+u,1,1+u^{2},1+u+u^{2},0,0,1+u+u^{2},u+u^{2}\right) $ & $%
1+u+u^{2}$ & $0$ & $211$ \\ \hline
\end{tabular}
\end{center}

\bigskip

In a similar way, as an application of Theorem \ref{ext}, 74 and 41 new codes
are obtained respectively from $\mathcal{C}_{11}\left(
0,u^{2},1+u^{2}\right) $ and $\mathcal{C}_{11}\left( u+u^{2},1+u,u\right) $.
For the codes which are not listed here the necessary information is available
online in \cite{kaya}. Hence $223$ codes in $W_{72,1}$ are obtained which
have new weight enumerators as; $\gamma =9$ and $\beta =$311, $\gamma =8$
and $\beta =$277, 291, $\gamma =7$ and $\beta =$262, 278, 280, 287, 296 , $%
\gamma =6$ and $\beta =$253, 255, 261, 263, 267, 275, 283, 285, 305, $\gamma
=5$ and $\beta =$228, 229, 231, 234, 235, 236, 242, 249, 255, 259, 265, 266,
269, 273, .., 278, 281, 283, 285, 286, 288, $\gamma =4$ and $\beta =$ 229,
231, 245, 249, 253, 259,\ 263, 264, 266, 273, 275, 279, 287, 292, $\gamma =3$
and $\beta =$196, 210, 215, 217, 218, 219, 231, 236, 238, 241, 244, 245,
248, 250, 251, 252, 254, 256, 258, 260, 261, 262, 266, 267, 268, 270, 272,
273, 276, 280, 284, 294, 297, $\gamma =2$ and $\beta =$195, 199, 201, 218,
219, 222, 223, 228, 231, .., 233, 239, 240, 241 243, 244, 245, 250, 251,
255, 257, 261, 262, 264, 266, 267, 268, 276, 278, 279, 285, $\gamma =1$ and $%
\beta =$193, 195, 199, 200, 206, 207, 208, 211, 212, 213, 215, 216, 217,
219, 220, 222, 223, 225, 226, 227, 229, 232, .., 240, 242, 243, 244, 246,
247, 248, 249, 250, 252, 254, 256, 257, 258, 260, 261, 264, 266, 270, 274,
276, 277, $\gamma =0$ and $\beta =$ 185, 196, 200, 203, 205, ..., 218, 220,
221, 222, 226, 227, 228, 230, 231, 232, 233, 234, 235, 237, 238, 239, 242,
.., 249, 251, 254, 257, 258, 261, 262, 264, 265, 267, 273, 275, 279.

\subsection{New binary extremal codes of length $68$}

There are two possibilities for the weight enumerators of extremal self-dual
$\left[ 68,34,12\right] _{2}$ codes (\cite{dougherty1}):
\begin{eqnarray*}
W_{68,1} &=&1+\left( 442+4\beta \right) y^{12}+\left( 10864-8\beta \right)
y^{14}+\cdots , \\
W_{68,2} &=&1+\left( 442+4\beta \right) y^{12}+\left( 14960-8\beta
-256\gamma \right) y^{14}+\cdots
\end{eqnarray*}%
where $\beta $ and $\gamma $ are parameters. Tsai et al. constructed a
substantial number of codes in both possible weight enumerators in \cite%
{tsai}. Most recently, $28$ new codes including the first examples with $%
\gamma =4$ and $\gamma =6$ in $W_{68,2}$ are obtained in \cite{karadeniz}.
For the list of codes with $\gamma =4$ and $\gamma =6$ in $W_{68,2}$ we
refer to \cite{karadeniz}. Together with the ones in \cite{karadeniz} codes
exists for $W_{68,2}$ when $\gamma =0$ and $\beta =$38, 40, 44, 45,
47,...,65,67,...,110,130, 132, 136, 138, 170, 204, 238, 272 or $\beta \in
\left\{ 2m|56\leq m\leq 62\right\} $; $\gamma =1$ and $\beta =$61, 63, 64,
65, 72, 73, 76, 82,\ldots , 115; and $\gamma =2$ with $\beta =$65, 71, 77,
86, 88, 93, 94, 96, 99, 109, 123, 130, 132, 134, 140, 142, 146, 152 or $%
\beta \in \left\{ 2m|51\leq m\leq 63\right\} $. For a list of known codes in
$W_{68,1}$ we refer to \cite{tsai}.

In the following we apply the extension method in Theorem \ref{ext} to the
binary images of the QDC codes $\mathcal{C}_{11}\left(
0,u^{2},1+u^{2}\right) $ and $\mathcal{C}_{11}\left( u+u^{2},1+u,u\right) $
and obtain $32$ new extremal self dual codes in $W_{68,2}$, codes with these
weight enumerators were not known to exist previously. In the following
tables, $C_{i}$ is the binary code generated by
\begin{equation*}
\left[
\begin{array}{cc|ccc}
1 & 0 &  & X &  \\ \hline
y_{1} & y_{1} & \multicolumn{3}{c}{} \\
\vdots  & \vdots  &  & G &  \\
y_{33} & y_{33} &  &  &
\end{array}%
\right]
\end{equation*}%
where $y_{i}=\left\langle G_{i},X\right\rangle $ for $1\leq i\leq 33$, $G$
is the matrix $\varphi \left( C_{11}\left( 0,u^{2},1+u^{2}\right) \right) $
and $\varphi \left( C_{11}\left( u+u^{2},1+u,u\right) \right) $ respectively
for tables 7 and 8. In order to save space the necessary vectors for
extensions are given in hexadecimal form, the binary vectors are available
online in \cite{kaya}.

\begin{center}
\bigskip

Table 7: New extremal self dual $\left[ 68,34,12\right] _{2}$ codes from $%
\mathcal{C}_{11}\left( 0,u^{2},1+u^{2}\right) $

\begin{tabular}{|l|l|l|l|l|}
\hline
& $X$ (hexadecimal) & $\gamma $ & $\beta $ & $\left\vert Aut\right\vert $ \\
\hline
$C_{1}$ & 1366E7855836D5F97 & $0$ & $111$ & $1$ \\ \hline
$C_{2}$ & 152C8FDA100E589E4 & $0$ & $113$ & $1$ \\ \hline
$C_{3}$ & 307C91A5CC0BEFB39 & $0$ & $115$ & $1$ \\ \hline
$C_{4}$ & 2FBF977F66C73C095 & $0$ & $117$ & $1$ \\ \hline
$C_{5}$ & 2DBBF3D2D8C219910 & $0$ & $119$ & $1$ \\ \hline
$C_{6}$ & 252951E0B1E5AAC21 & $0$ & $121$ & $1$ \\ \hline
$C_{7}$ & EDA2BBD6B53937A4 & $0$ & $123$ & $1$ \\ \hline
$C_{8}$ & 4528892715B1C268 & $0$ & $125$ & $1$ \\ \hline
$C_{9}$ & D989EFC395464C6F & $0$ & $126$ & $1$ \\ \hline
$C_{10}$ & 42E4E15D93AE3075 & $0$ & $127$ & $1$ \\ \hline
$C_{11}$ & DC2E97A7B77B9378 & $0$ & $128$ & $1$ \\ \hline
$C_{12}$ & 20C589DC55E710589 & $0$ & $129$ & $1$ \\ \hline
$C_{13}$ & 22C125C827448086F & $0$ & $131$ & $1$ \\ \hline
$C_{14}$ & 231CC8E70F78AE4F0 & $0$ & $133$ & $1$ \\ \hline
$C_{15}$ & 32BC23AA33E36B123 & $0$ & $134$ & $1$ \\ \hline
$C_{16}$ & 26745142F8B420C86 & $0$ & $135$ & $2$ \\ \hline
$C_{17}$ & 38C21CF4AF47A41E3 & $0$ & $139$ & $2$ \\ \hline
$C_{18}$ & 384F6537649B8B0AA & $1$ & $118$ & $1$ \\ \hline
$C_{19}$ & 6353300D871453E1 & $1$ & $126$ & $1$ \\ \hline
$C_{20}$ & CE66C92ABB5EE18E & $1$ & $129$ & $1$ \\ \hline
$C_{21}$ & 739A837C7816DDCE & $1$ & $132$ & $1$ \\ \hline
$C_{22}$ & 190A5C0A051314F9B & $1$ & $133$ & $1$ \\ \hline
$C_{23}$ & 25F97FDA3C7DD9F16 & $1$ & $138$ & $1$ \\ \hline
$C_{24}$ & 3DB29DEB3DFDA30C1 & $1$ & $140$ & $2$ \\ \hline
$C_{25}$ & 3BFBD24B7741E669F & $1$ & $142$ & $1$ \\ \hline
$C_{26}$ & 18DAFB91A9516B39 & $1$ & $146$ & $1$ \\ \hline
\end{tabular}

\bigskip

Table 8: New extremal self dual $\left[ 68,34,12\right] _{2}$ codes from $%
\mathcal{C}_{11}\left( u+u^{2},1+u,u\right) $

\begin{tabular}{|l|l|l|l|l|}
\hline
& $X$ (hexadecimal) & $\gamma $ & $\beta $ & $\left\vert Aut\right\vert $ \\
\hline
$C_{27}$ & E2A99BBA87FEF283 & $0$ & $66$ & $1$ \\ \hline
$C_{28}$ & 289CF22D186686C0E & $1$ & $77$ & $2$ \\ \hline
$C_{29}$ & 14AD41A72715F3696 & $1$ & $79$ & $2$ \\ \hline
$C_{30}$ & 2C8C98C94932D7341 & $1$ & $81$ & $1$ \\ \hline
$C_{31}$ & 3D07A44D2980F9E8C & $2$ & $82$ & $1$ \\ \hline
$C_{32}$ & 3E26AD3A8670694F8 & $2$ & $84$ & $2$ \\ \hline
\end{tabular}

\bigskip
\end{center}

In addition to these codes we were able to find codes in $W_{68,2}$ with
automorphism group or order $2$ with $\gamma =0$ and $\beta =$66, 113, 117,
119, 121, 123, 125, 126, 127, 128, 129, 133, 134 and codes which have
automorphism group of order $4$ with $\gamma =0$, $\beta =128$ and $\gamma
=1 $, $\beta =146$ . We do not list these $15$ codes here, they are
available online at \cite{kaya}.

\section{Conclusion}

Quadratic residue codes have been of interest to the coding theory community
because of their algebraic structures and their potential to construct good
codes. As illustrated by their role in constructing the extremal $[48,24,12]$
Type II code, they can also be of help in constructing self-dual codes. We
considered quadratic residue codes over a specific Frobenius ring that is
endowed with a duality and distance preserving Gray map. Using different
constructions for self-dual codes over $R$ we were able to obtain many new
extremal binary self-dual codes as Gray images. Because of the automorphisms
resulting from the ring structure as well as the quadratic residue
structure, our constructions have high potential to fill the gaps in the
literature on self-dual codes.

As a possible line of research, different rings can be considered for
similar constructions.

\end{document}